\documentclass[sigconf, anonymous=false]{acmart}
\usepackage{algorithm}
\usepackage{algorithmic}
\usepackage{amsfonts}
\usepackage{amsthm}
\usepackage{balance}
\usepackage[title]{appendix}
\newtheorem*{proposition}{Proposition}

\DeclareMathOperator*{\argmax}{arg\,max}

\AtBeginDocument{%
  \providecommand\BibTeX{{%
    \normalfont B\kern-0.5em{\scshape i\kern-0.25em b}\kern-0.8em\TeX}}}

\copyrightyear{2023} 
\acmYear{2023} 
\setcopyright{acmlicensed}\acmConference[AIES '23]{AAAI/ACM Conference on
AI, Ethics, and Society}{August 8--10, 2023}{Montréal, QC, Canada}
\acmBooktitle{AAAI/ACM Conference on AI, Ethics, and Society (AIES '23),
August 8--10, 2023, Montréal, QC, Canada}
\acmPrice{15.00}
\acmDOI{10.1145/3600211.3604660}
\acmISBN{979-8-4007-0231-0/23/08}

\begin{document}

\title{Mitigating Voter Attribute Bias for Fair Opinion Aggregation}

\author{Ryosuke Ueda}
\orcid{0009-0001-3892-9857}
\affiliation{%
  \institution{Kyoto University}
  \city{Kyoto}
  \country{Japan}
}
\email{rueda@ml.ist.i.kyoto-u.ac.jp}

\author{Koh Takeuchi}
\affiliation{%
  \institution{Kyoto University}
  \city{Kyoto}
  \country{Japan}
}
\email{takeuchi@i.kyoto-u.ac.jp}

\author{Hisashi Kashima}
\affiliation{%
  \institution{Kyoto University}
  \city{Kyoto}
  \country{Japan}
}
\email{kashima@i.kyoto-u.ac.jp}


\begin{abstract}
The aggregation of multiple opinions plays a crucial role in decision-making, such as in hiring and loan review, and in labeling data for supervised learning. 
Although majority voting and existing opinion aggregation models are effective for simple tasks, they are inappropriate for tasks without objectively true labels in which disagreements may occur. 
In particular, when voter attributes such as gender or race introduce bias into opinions, the aggregation results may vary depending on the composition of voter attributes. 
A balanced group of voters is desirable for fair aggregation results but may be difficult to prepare. 
In this study, we consider methods to achieve fair opinion aggregation based on voter attributes and evaluate the fairness of the aggregated results.

To this end, we consider an approach that combines opinion aggregation models such as majority voting and the Dawid and Skene model (D\&S model) with fairness options such as sample weighting. 
To evaluate the fairness of opinion aggregation, probabilistic soft labels are preferred over discrete class labels. 
First, we address the problem of soft label estimation without considering voter attributes and identify some issues with the D\&S model. 
To address these limitations, we propose a new Soft D\&S model with improved accuracy in estimating soft labels. 
Moreover, we evaluated the fairness of an opinion aggregation model, including Soft D\&S, in combination with different fairness options using synthetic and semi-synthetic data. 
The experimental results suggest that the combination of Soft D\&S and data splitting as a fairness option is effective for dense data, whereas weighted majority voting is effective for sparse data. 
These findings should prove particularly valuable in supporting decision-making by human and machine-learning models with balanced opinion aggregation.
\end{abstract}




\maketitle

\section{INTRODUCTION}
Real-world decision-making processes such as recruitment, loan approval, and elections often require aggregations of opinions from multiple stakeholders such as interviewers, bankers, and the general public. 
Aggregating opinions on simple and objective questions such as determining the presence of a car in an image is relatively straightforward; this is often the case with supervised learning from crowdsourced labels~\cite{Snow2008-qo, Deng2009-ak, Ipeirotis2010-sr, Liu2012-iw, Rodrigues2013-vl, Bowman2015-xm}.
However, disagreements often occur, particularly in questions that rely on the subjective judgments of respondents, in which ground truth answers do not exist. 
Voters have different backgrounds and perspectives, which influence their evaluations and lead to disagreements and differences in opinions~\cite{Akhtar2020-qw, Kumar2021-lq, Chen2021-ir, Rottger2022-nt}.
This discrepancy can be further exacerbated by voter attribute bias, which is a bias in a set of opinions that depend on voter attributes such as gender and race resulting in biased aggregation results~\cite{Biswas2020-me, Kumar2021-lq, Sap2022-cg}.

Although a well-balanced panel of voters is ideal to fairly aggregate the opinions of various segments of the population, maintaining such a balanced composition is a major challenge.
The recent development of decision support methods based on prediction using artificial intelligence has attracted considerable attention~\cite{Duan2019-pw, Green2019-vy}, and raised some concerns about the possibility of social disadvantage resulting from unfair predictions caused by voter attribute bias.
Several studies have examined fairness in opinion aggregation~\cite{Li2020-iu, Biswas2020-me}, and a recent work has considered fairness with respect to voter attributes \cite{Gordon2022-yy}.
Therefore, in this study, we propose methods to fairly aggregate opinions from an unbalanced panel of voters, and a procedure to evaluate the fairness of the aggregation results by considering the degree of disagreement and the voter attributes.

To achieve fair opinion aggregation, we first consider models for subjective opinion aggregation.
Several aggregation models are well known, including majority voting and the Dawid and Skene model (D\&S model)~\cite{Dawid1979-fq}.
In cases without any definitive correct answer, considering the aggregation result (often treated as a latent true label in opinion aggregation models) as a soft label, which represents the proportion of opinions, would be preferable.
We show some limitations of the D\&S model in estimating soft labels and propose a novel Soft D\&S model that explicitly considers soft ground truth labels.
In addition, to address attribute bias, we combine opinion aggregation models with three fairness options, including sample weighting, data splitting, and GroupAnno~\cite{Liu2022-vk}.

We evaluated the fairness of various opinion aggregation models in combination with fairness options, using both synthetic and semi-synthetic data derived from real-world data. 
The experimental results indicate that combining of Soft D\&S and data splitting is effective for dense data, whereas weighted majority voting is suitable for sparse data.
This result could be attributed to the fact that Soft D\&S requires a parameter for each voter, which in turn requires a large enough dataset to estimate these parameters accurately.

The key contributions of this study are summarized as follows.
\begin{itemize}
    \item To the best of our knowledge, the present work is the first to propose methods to aggregate opinions fairly in terms of voter attributes and evaluate them quantitatively.
    \item We propose a new Soft D\&S model, an extension of D\&S, which addresses the issue of sharp output in the D\&S model and improves the estimation accuracy of soft labels.
\end{itemize}

\section{PROBLEM SETTING}
First, we formulate the general problem setting for opinion aggregation~\cite{Zhang2016-ql, Li2016-ux, Zheng2017-tp}.
Consider a group of human voters and a set of tasks that require appraisal, indexed as voter $i = 1, 2, \ldots, I$, and task $j = 1, 2, \ldots, J$, respectively.
Here, a task refers to an entity that requires evaluation, such as a single image in annotations used for image classification or a single applicant in recruitment.
Given that we focus on opinion aggregation with multi-class labels in this work, task $j$ is labeled with a $K$-class by multiple voters.
Let $X_{ij} \in \{-1, 1, 2, \ldots, K\}$ denote the label assigned by voters $i$ to task $j$, and let the $I \times J$ matrix with $X_{ij}$ as an element be denoted as $X$.
Note that a voter is not obliged to label all tasks, and $X_{ij} = -1$ for $(i, j)$ pairs where no label is provided.
In the general opinion aggregation problem, a discrete $K$-class label is assumed as the true label for each task. 
For example, in the task of determining whether a car is present in given image, the problem assumes that each task involves a possible binary label of "Yes" or "No".

Up to this point, we have presented the general problem setup for opinion aggregation. 
In this study, we introduce two changes particularly to address fairness concerns related to voter attributes.
The first change is that instead of assuming a discrete class label as the true label for each task, we assume continuous soft labels to handle more complex tasks in which disagreement among voters may be expected. 
The true soft label for each task $j$ is denoted by $Z_j = (Z_{j1}, Z_{j2}, \ldots, Z_{jK})^\top$, where $Z_{jk} \in [0, 1]$ represents the degree to which task $j$ belongs to class $k$. 
Because $Z_j$ is a soft label, it satisfies the constraints $\sum_{k=1}^K Z_{jk} = 1$ for all $j$. 
Note that the input $X_{ij}$ is a discrete label as same as the general setting.

The second change is that inputs include the representation of each voter attributes such as gender and race.
In particular, each voter takes a binary attribute $a_i \in \{0, 1\}$. 
Although considering more complex voter attributes would be preferable, we focus on a single binary attribute in this study for simplicity.
For tasks in which opinions are conflicted, a bias may be present in opinions due to voter attributes, which we refer to as voter attribute bias in this study.
Traditional methods for opinion aggregation tend to assign more weight to the opinions of a majority group of voters, leading to the aggregate results being dominated by the attributes of the majority of voters even though their opinions may be influenced by voter attribute bias.
In an ideal scenario, to ensure fairness in the aggregation, a balanced group of voters should employed in terms of gender, race, and other relevant attributes. 
However, assembling such a balanced group of voters is often challenging in practice.
Therefore, our goal is to estimate the aggregate results of the opinions of a ideally balanced group of voters, which are not directly observable in the real-world, from the opinions of an unbalanced group of voters as shown in Figure~\ref{fig:ideal_voters}.

To formalize the problem setup, let $p(a)$ denote the distribution of voter attributes for the ideal group. 
For example, in the case of binary attributes such as role (representing students as $a=0$ and teachers as $a=1$), if the ideal group comprises equal numbers of students and teachers, then $p(a=0) = p(a=1) = 0.5$.
In this study, the true soft label $Z$ is defined as a soft label determined by majority voting when a sufficient number of voters whose voter attribute distribution follows $p(a)$ are present. 
However, in practice, the actual observed voter population may not follow $p(a)$, and the number of voters may be limited. 
Thus, in this study, we aim to estimate $Z$ from the input label $X$ and voter attributes $\{a_i\}_{i=1}^I$.

\begin{figure}
    \centering
    \includegraphics[width=\linewidth, bb=0 356 594 520]{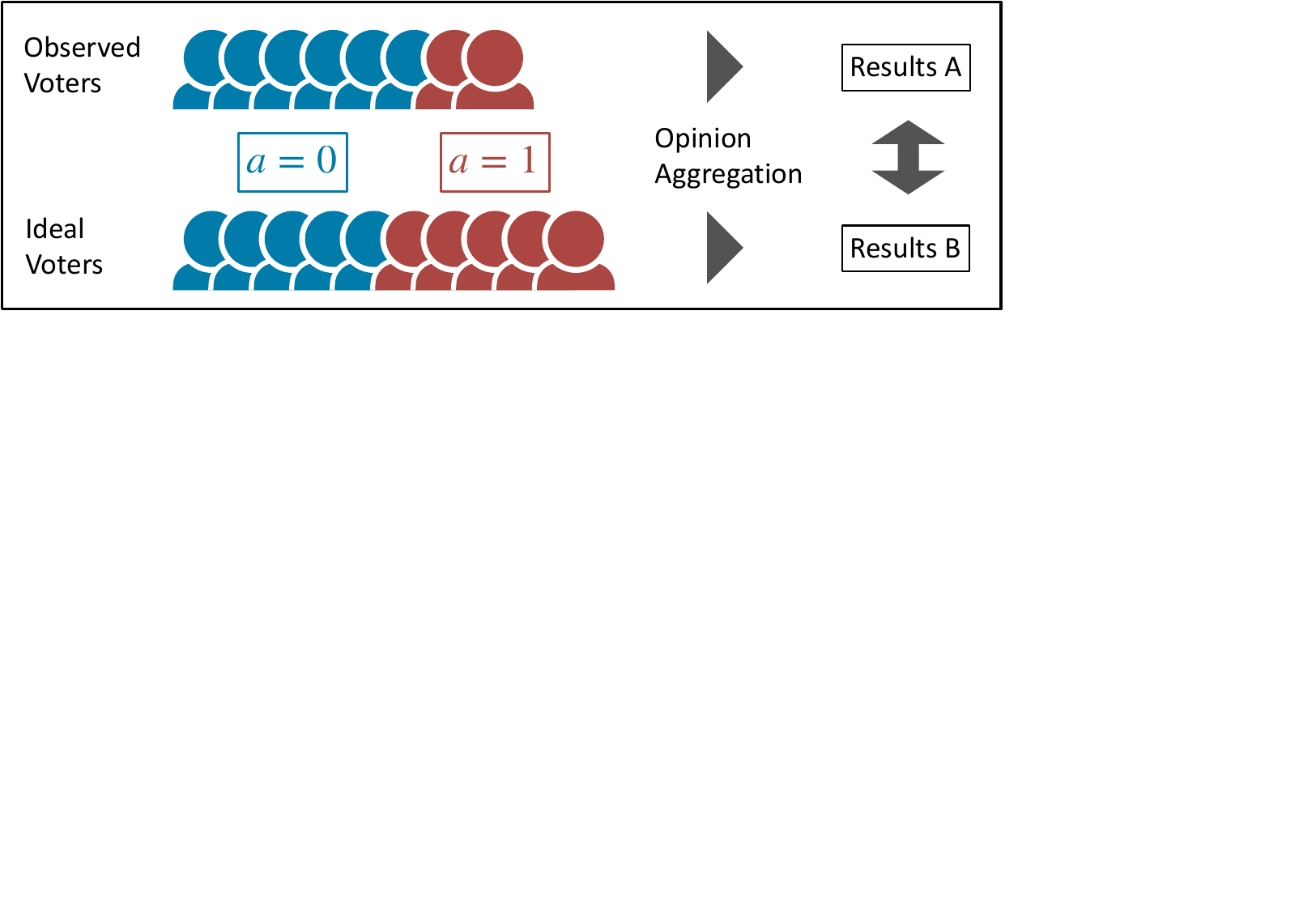}
    \caption{Our goal is to perform fair opinion aggregation with respect to voter attribute $a\in \{0,1\}$, i.e., opinion aggregation in an ideal population that has equal numbers of voters with $a=0$ and $a=1$. When the distribution of the observed voter attribute deviates from that of the ideal population, there can be a systematic discrepancy between opinion aggregation results A and B. We want to estimate the fair result B from observed voter labels alone.}
    \label{fig:ideal_voters}
\end{figure}
\section{RELATED WORK} \label{sec:related}
\subsection{Voter Attribute in Opinion Aggregation}
Several existing studies have examined voter attributes in the context of opinion aggregation. 
First, Kazai et al.~\cite{Kazai2012-ab} investigated the relationship between voter attributes such as location, gender, and personality traits and the quality of labels in crowdsourcing. 
They found a strong correlation between label quality and the geographic location of voters, particularly those located in the United States, Asia, and Europe. 
Second, Liu et al.~\cite{Liu2022-vk} proposed a model called GroupAnno which incorporated voter attributes into an opinion aggregation framework. 
Their work addressed the issue of estimating parameters for voters with limited responses and improved the accuracy of aggregation results. 
In the present study, we draw inspiration from GroupAnno to enhance fairness with respect to voter attributes, rather than improving accuracy. 
GroupAnno is originally based on the Learning From Crowds model (which is derived from the D\&S model~\cite{Dawid1979-fq}); as discussed in Section~\ref{subsubsec:DS}, the D\&S model suffers from problems with soft label estimation. 
In addition to evaluating fairness, we address the problem of soft label estimation using GroupAnno.

Another study of interest also explored fairness in opinion aggregation through the use of voter attributes. 
Gordon et al.~\cite{Gordon2022-yy} investigated the problem of opinion aggregation with a focus on social minority voters. 
They utilized an annotated dataset~\cite{Kumar2021-lq} that included voter attributes such as race, gender, age, and political attitudes to measure the toxicity of social media comments. 
Their study considered more complex voter attributes than the present study, including multiple pairs of attributes such as race (including Hispanic and Native Hawaiian), gender (including non-binary), and political attitudes. 
They first trained a deep learning regression model designed to consider the textual features of comments, voter attributes, and voter IDs to estimate a five-level toxicity label. 
Using this model, they generated toxicity labels for any comment made by a virtual voter with arbitrary voter attributes, including social minorities, and aggregated their opinions. 
While deep learning models have high expressive power and can handle complex voter attributes, there is a concern that the labels are generated by less interpretable models instead of humans.
In contrast, we propose a novel approach based on a traditional opinion aggregation model that does not rely on text features of tasks and is relatively more interpretable.
The experiments by Gordon et al. focus on the accuracy of estimating $X_{ij}$, the labels for each voter, whereas we directly assess the fairness of the aggregation results $Z_j$ by considering the balanced or unbalanced attributes of voters.

\subsection{General Opinion Aggregation Models}
Opinion aggregation has become a significant area of research  with the advent of crowdsourcing platforms such as Amazon Mechanical Turk and the growing need for labeling in machine learning. 
One of the main challenges in opinion aggregation is that of ensuring quality control, because voters are human~\cite{Li2016-ux}. 
This challenge is particularly acute when labeling is outsourced through crowdsourcing, where assessing the ability and motivation of voters is more difficult due to the online nature of the process, which leads to considerable variability in the quality of the generated labels. 
To address this issue, numerous opinion aggregation models have been proposed to capture variance in label quality~\cite{Zheng2017-tp}.

Dawid and Skene proposed an opinion aggregation model that utilizes a confusion matrix to model voters~\cite{Dawid1979-fq}. 
The D\&S model applies an EM algorithm to iteratively optimize the voter confusion matrix and the true labels. 
Further details about the model are provided in Section~\ref{subsubsec:DS}. 
Several opinion aggregation models based on the D\&S model have been introduced to date. 
In this study, we present the most representative models. 
The Learning From Crowds (LFC) model~\cite{Raykar2010-io} learns a classifier with task features and voter labels as input and can also be used as an opinion aggregation model when task features are not available. 
In this case, the model is an extension of the D\&S model that maximizes the posterior probability by introducing a Dirichlet prior distribution for the confusion matrix and true label estimates. 
In contrast, the Bayesian Classifier Combination (BCC) model~\cite{Kim2012-wh} aggregates multiple classifiers and can be considered an opinion aggregation model when the classifiers are replaced with human voters. 
In~\cite{Kim2012-wh}, Kim and Ghahramani proposed Independent BCC (IBCC) and Dependent BCC (DBCC) assuming the opinions of independent and correlated, respectively. 
IBCC extends the D\&S model to Bayesian estimation and introduces a Dirichlet prior distribution for the confusion matrix and true label estimates, similarly to LFC. 
Community BCC (CBCC)~\cite{Venanzi2014-fx} model is designed to address the ineffectiveness of IBCC for cases in which labels are scarce, and it extends IBCC by grouping similar voters. 
Bayesian estimation is conducted using the expectation propagation method, assuming a graphical model in which each voter belongs to a single group and the confusion matrices of voters in the same group have similar values. 
Due to the high computational cost of the DBCC when the number of voters is large, Enhanced BCC (EBCC)~\cite{Li2019-ebcc} was developed to reduce computational complexity and incorporate correlation among voters in the model.

Some alternative approaches to opinion aggregation models that do not use a confusion matrix have also been proposed.
ZenCrowd~\cite{Demartini2012-gy} uses the percentage of correct responses per voter as a real number in the interval $[0, 1]$ rather than a confusion matrix. 
The correct response rate and true label per task are estimated using the EM algorithm.
GLAD~\cite{Whitehill2009-rk} was inspired by item response theory~\cite{Linden1997-rc} and models the ability of a voter $i$ and the difficulty of a task $j$ with one-dimensional parameters $\alpha_i$ and $\beta_j$, respectively. 
They assume the probability that $X_{ij}$ matches the true label to be $\sigma(\alpha_i \beta_j)$ using the sigmoid function, and perform maximum likelihood estimation using the EM algorithm.
The model by Zhou et al.~\cite{Zhou2012-iw} assumed a probability distribution of labels for each pair of voter and task. 
It modeled not only the ability of voters but also the difficulty of tasks and could also represent the interaction between voters and tasks.
Bayesian Weighted Average (BWA)~\cite{Li2019-bwa} assumes a normal distribution for the process of generating discrete binary labels. 
The label $X_{ij}$ is assumed to follow $\mathcal{N}(z_j, v_i^{-1})$, and $z_j, v_i$ are optimized in the framework of Bayesian inference.
The aggregation result is 1 if $z_j$ is greater than 0.5, and 0 otherwise.
It can also be extended to multi-class classification.

\subsection{Fairness in Opinion Aggregation}
Recently, the focus on fairness in machine learning has been increasing, particularly in opinion aggregation models that are commonly used to generate training data. 
While our study addressed the issue of fair opinion aggregation with respect to voter attributes, Li et al.~\cite{Li2020-iu} addressed fairness with respect to task attributes in cases where the task is performed by a human being. 
In particular, they investigated fairness with respect to gender and race of defendants in the United States in the context of a recidivism prediction task using the publicly available dataset~\cite{Dressel2018-jq}. 
In their work, they employed Statistical Parity~\cite{Dwork2012-ie} as a fairness measure, which is often used for fairness in classification problems, and proposed an opinion aggregation model that incorporated such constraints to prevent unfairly high or low labeling of recidivism risk based on defendant attributes. 
However, our study differs significantly in that it focuses on fairness with respect to voter attributes rather than tasks attributes. 

Notably, some studies have also explored modifying experimental designs to improve fairness in opinion aggregation. 
For example, in the recidivism prediction dataset for U.S. defendants~\cite{Dressel2018-jq} mentioned earlier, a subset of 1,000 individuals was randomly sampled from a larger dataset of 7,214 defendants. 
Biswas et al.~\cite{Biswas2020-me} took a similar approach and sampled 1,000 individuals from the same dataset, with 250 individuals for each of four groups, including African-American recidivists, African-American non-recidivists, Caucasian recidivists, and Caucasian non-recidivists. 
They then collected a new dataset with an equal number of black and white voters and used the Equalized Odds~\cite{Hardt2016-me} fairness measure to assess fairness with respect to task attributes. 
Their findings suggest that the voter labels were fairer in the newly created dataset than in the original dataset, and a classification model trained on the dataset with balanced defendant attributes was also fairer. 
While their study used voter attributes, their assessment of fairness was limited to the attributes of the task, i.e., the defendant.

\subsection{Soft Labels for Machine Learning}
Soft labels expressing uncertainty or disagreement among voters, can provide additional information and potentially enhance the accuracy of machine learning models~\cite{Uma2020-qb, Peterson2019-ll}. 
Multi-task learning in which soft label estimation is performed as an auxiliary task, has shown improved accuracy compared to models trained solely on hard labels in some natural language processing tasks~\cite{Fornaciari2021-pv}. 
Soft labels are particularly important in tasks where voter disagreement is expected, such as comment toxicity classification. 
Because hard labels are not suitable for evaluating such problems, Gordon et al.~\cite{Gordon2021-rq} proposed a method of sampling multiple hard labels with a soft label for each comment.
They also proposed a method to estimate soft labels using singular value decomposition to eliminate noise. 
Davani et al.~\cite{Davani2022-ry} demonstrated the usefulness of multi-task learning to estimate labels per voter using an annotated dataset for subjective tasks in natural language processing. 
They compared models trained on data previously aggregated into hard labels by majority voting to models trained by multi-task learning per annotator without opinion aggregation and found that the latter achieved equal or better accuracy. 
\section{PROPOSED METHODS} \label{sec:proposed}
To estimate unbiased soft labels, we combine the opinion aggregation model with the fairness option.
Opinion aggregation models take $X$ as input and produces a soft label $\hat{Z}$ as output. Some examples of opinion aggregation models include Majority Voting (MV) and the D\&S model, which are described below.
Because the input of the opinion aggregation model does not include voter attributes, the resulting soft labels obtained from this model alone are not unbiased.
First, we identify a problem in soft label estimation using D\&S and propose an extension of D\&S called Soft D\&S that addresses this problem.

The fairness option is a method to increase fairness in combination with an opinion aggregation model. 
In this study, we adopt three fairness options, including sample weighting, data splitting, and GroupAnno. 
The fairness of each pair of an opinion aggregation model and a fairness option presented in this section were verified through experiments as described in Section \ref{sec:experiments}.

\subsection{Opinion Aggregation Models} \label{subsec:agg_models}
As mentioned earlier, we first discuss opinion aggregation models designed to estimate soft labels without considering fairness with respect to voter attributes. 
We introduce the simplest opinion aggregation model MV, and then introduce the D\&S model, which takes voter reliability into account. 
We then identify a problem with the ability of the D\&S model to estimate soft label accurately in certain situations and propose a modified version of the model called ``Soft D\&S'' that addresses this issue.

\subsubsection{Majority Voting (MV)}
MV is a simple model that computes the ratio of labels assigned to each class by the voters, which is then directly output as a soft label. 
For example, consider a binary classification task $j$ in which 6 voters assign class 1 and 4 voters assign class 2. 
The soft label estimated by majority voting is $\hat{Z}_j = (0.6, 0.4)^\top$. 
This estimate can be formulated as follows:
\[
\hat{Z}_{jk} = \frac{\sum_i I(X_{ij} = k)}{\sum_i I(X_{ij} \neq -1)} .
\]
Figure \ref{fig:graphical_MV} shows the graphical model of MV, where 
\begin{equation}
p(X_{ij} \mid Z_j) = \mathrm{Categorical}(Z_j) \label{eq:mv}.
\end{equation}
Note that $\mathrm{Categorical}(\cdot)$ is a categorical distribution, which coincides with the Bernoulli distribution when $K=2$. 
To summarize, MV is an algorithm to estimate the parameter $Z_j$ of the categorical distribution assuming the graphical model represented in Figure \ref{fig:graphical_MV} and Equation (\ref{eq:mv}).

\begin{figure}
    \centering
    \includegraphics[width=0.4\linewidth, bb=0 200 260 360, clip=true]{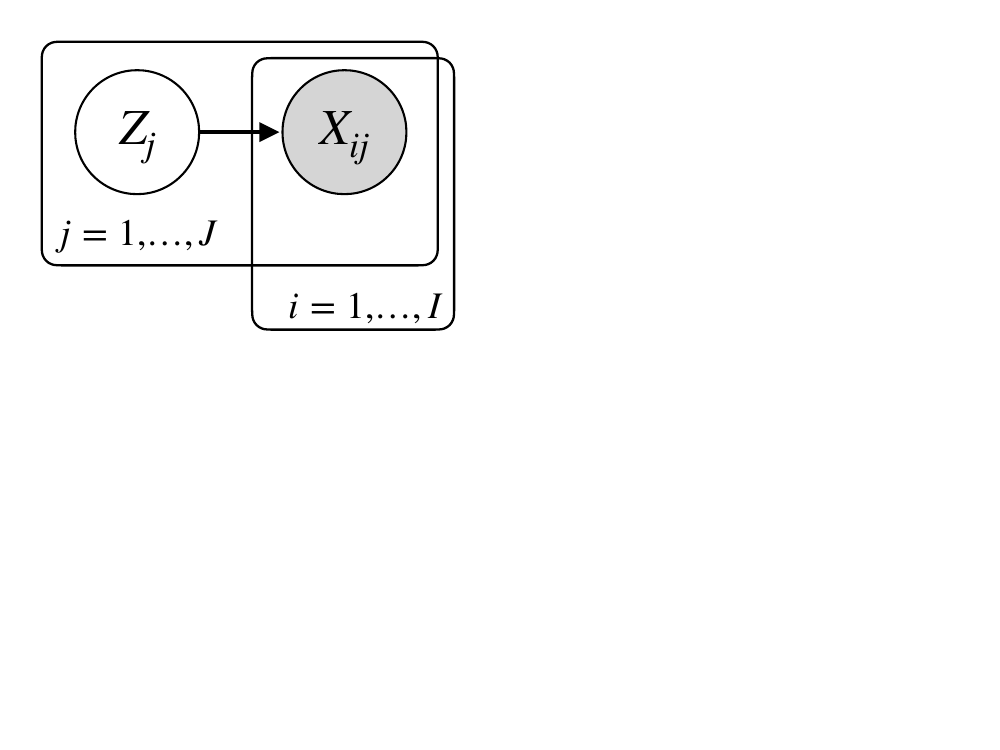}
    \caption{Graphical model of MV. Only shaded variables are observed.}
    \label{fig:graphical_MV}
\end{figure}

\subsubsection{Dawid and Skene Model (D\&S)} \label{subsubsec:DS}
\begin{figure}
    \centering
    \includegraphics[width=0.65\linewidth, bb=0 346 497 540, clip=true]{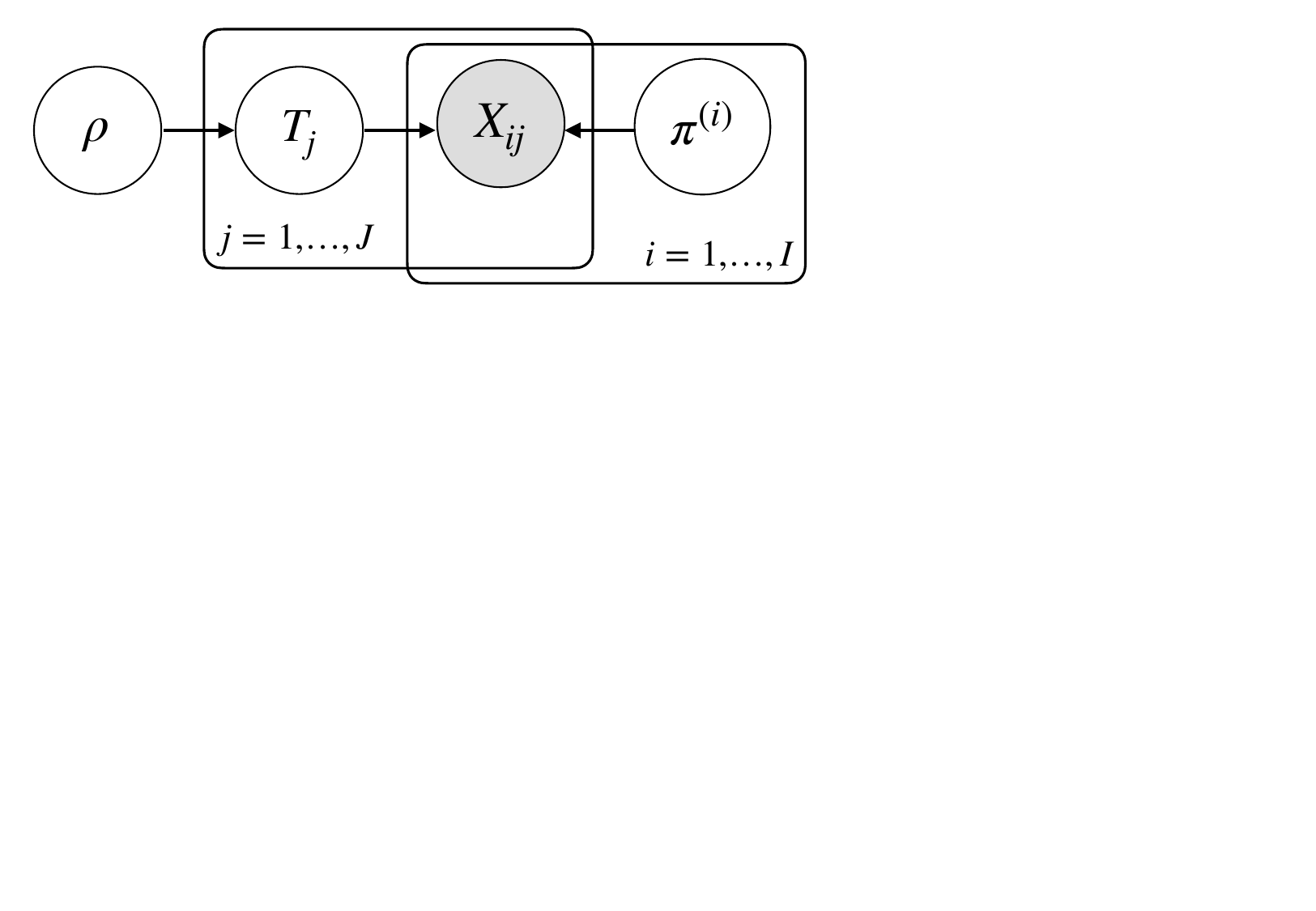}
    \caption{Graphical model of D\&S. Only shaded variables are observed.}
    \label{fig:graphical_DS}
\end{figure}

We then introduce the D\&S model~\cite{Dawid1979-fq}, which is a more sophisticated approach to opinion aggregation than MV. 
The D\&S model incorporates a confusion matrix for each voter, which is optimized using an EM algorithm. 

Figure \ref{fig:graphical_DS} shows the graphical model of D\&S, where $T_j$ is the true label of task $j$. 
Notably, in the D\&S model, $T_j$ assumes discrete labels, meaning each task $j$ has only one class label $T_j \in \{1, \ldots, K\}$. 
The confusion matrix for each voter $i$ is denoted as $\pi^{(i)} \in \mathbb{R}^{K \times K}$. 
In particular, for any $k, l \in \{1, \ldots, K\}$ and $j \in \{1, \ldots, J\}$, the confusion matrix element is defined as
\[\pi^{(i)}_{kl} = p(X_{ij} = l \mid T_j = k).\] 
For example, the confusion matrix of the best voter is $\pi^{(i)} = E_K$ (where $E_n$ refers to the $n \times n$ identity matrix), and this voter always labels the true class. 
In contrast, the confusion matrix of a random voter has all elements $1/K$.
Furthermore, a parameter $\rho = (\rho_1, \ldots, \rho_K)^\top$ represents the prior distribution such that $T_j \sim \mathrm{Categorical}(\rho)$ for any $j$.

Based on the assumptions made up to this point, a lower bound for the log-likelihood $\mathcal{L}$ can be derived when $X$ is observed, as given by the following inequality:
\begin{align}
\mathcal{L} &= \ln p(X \mid \pi, \rho) \notag \\
&= \ln \sum_T p(T \mid \rho) p(X \mid T, \pi) \notag \\
&= \sum_j \ln \sum_k \frac{q(T_j = k)}{q(T_j = k)} \rho_k \prod_{i: O_{ij} = 1} p(X_{ij} \mid T_j, \pi^{(i)}) \notag \\
&\geq \sum_j \sum_k q(T_j = k) \ln \frac{\rho_k}{q(T_j = k)} \notag \\
&\quad + \sum_j \sum_kq(T_j = k) \sum_{i: O_{ij} = 1} \ln p(X_{ij} \mid T_j, \pi^{(i)}), \label{eq:dawid_lb}
\end{align}
where $q(T_j = k)$ represents an arbitrary distribution of the discrete latent variable $T_j$, which corresponds to the soft labels.

This lower bound is maximized using the EM algorithm. 
During the E-step, the parameters $\rho$ and $\pi$ are fixed, and $q(T_j=k)$ is updated to maximize the lower bound. 
During the M-step, $q(T_j=k)$ is fixed, and the parameters $\rho$ and $\pi$ are updated to maximize the lower bound. 
In the original D\&S model, after the EM algorithm converges, the discrete label $T_j$ is estimated by comparing the obtained $q(T_j=k)$ with a threshold value.

\subsubsection{Sharpness of D\&S Output}
D\&S can estimate soft labels by utilizing $X$ as input and generating $q(T_j=k)$ as output. Nonetheless, optimization using the EM algorithm may lead to the concentration of $q(T_j=k)$ around either 0 or 1, thus producing estimates with high sharpness. 
Figure \ref{fig:MV-vs-DS} demonstrates the experimental results on synthetic data and a comparison with those of MV.

\begin{figure}
    \centering
    \includegraphics[width=0.99\linewidth]{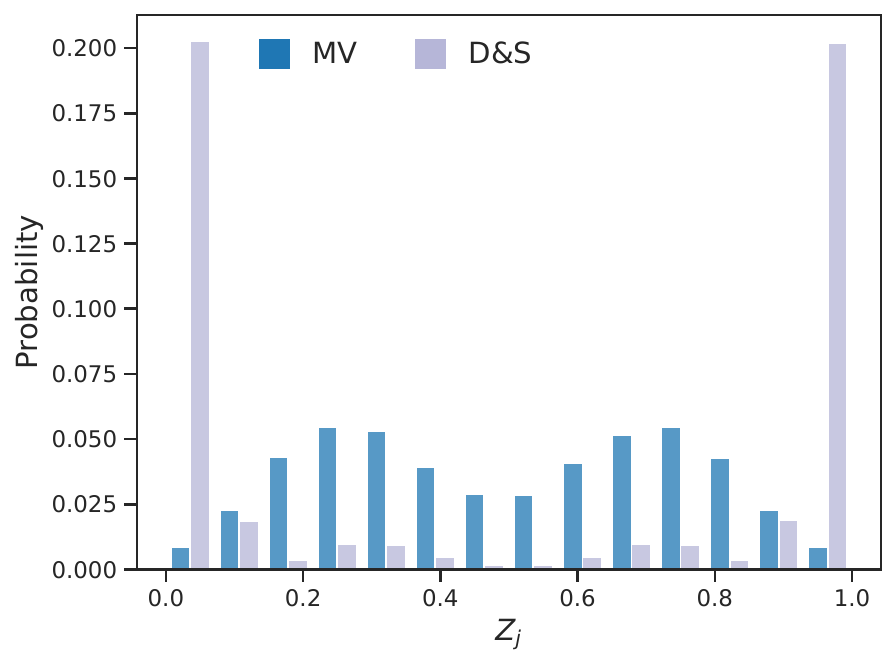}
    \caption{
    Soft label estimation results for D\&S and MV models. 
    We utilized synthetic data with $K=2$ classes, and 15 voters answering 1000 tasks. 
    The true soft label for 500 of the 1000 tasks was $(0.3, 0.7)^\top$ and $(0.7, 0.3)^\top$ for other 500 tasks. 
    We used the graphical model in Figure \ref{fig:graphical_MV} to generate the labels $X$. 
    The data generation and estimation was repeated 100 times.
    We show the distribution of estimated class 1 soft label, where D\&S tends to estimate a sharper distribution than MV.
    }
    \label{fig:MV-vs-DS}
\end{figure}

The EM algorithm produces sharp estimates due to the E-step, in which the update for $q(T_j=k)$ is defined as:
\[
q(T_j = k) \propto \rho_k \prod_{i, l} {\pi^{(i)}_{kl}}^{I(X_{ij} = l)}.
\]
To illustrate this issue, we consider a scenario, in which ten individuals vote on a single task in a binary classification task ($K=2$), and the confusion matrix for the D\&S model across all voters is
\[
\begin{bmatrix}
    0.9 & 0.1 \\
    0.1 & 0.9
\end{bmatrix}.
\]
For example, assuming that 6 out of 10 voters cast their votes for class 1 and the other 4 for class 2, with a confusion matrix of the D\&S model as previously mentioned, we can compute $q(T_j=1)$ and $q(T_j=2)$ using the E-step of the D\&S model (assuming that the prior distribution $\rho$ is uniformly distributed), as follows:
\begin{align*}
    q(T_j = 1) &= \frac{0.9^6 0.1^4}{0.9^6 0.1^4 + 0.9^4 0.1^6} \approx 0.988, \\
    q(T_j = 2) &= \frac{0.9^4 0.1^6}{0.9^6 0.1^4 + 0.9^4 0.1^6} \approx 0.012.
\end{align*}
The estimates obtained from the D\&S model are much sharper than the MV estimate $(0.6, 0.4)^\top$, which highlights the difficulty of the D\&S model in detecting voter attribute bias in scenarios where such discrepancies occur.

\subsubsection{Soft D\&S}
\begin{figure}
    \centering
    \includegraphics[width=0.6\linewidth, bb=0 137 331 360, clip=true]{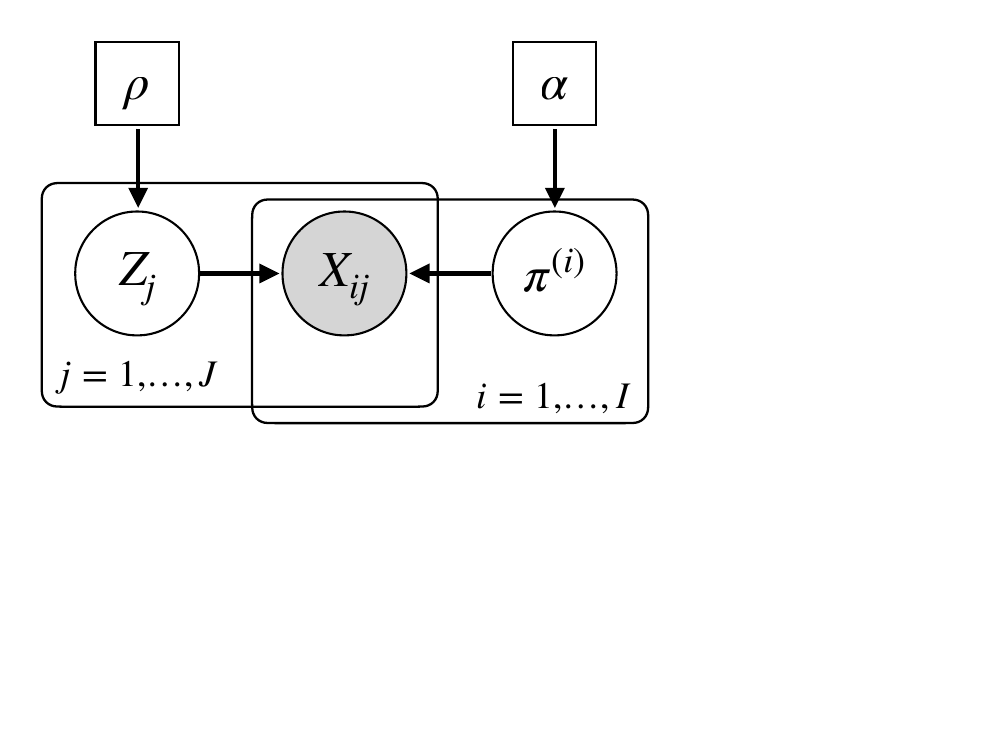}
    \caption{Graphical model of Soft D\&S. Only shaded variables are observed; variables surrounded by squares are hyperparameters.}
    \label{fig:graphical_SoftDS}
\end{figure}
To address the issue mentioned above, we propose a solution called the Soft D\&S model, which is an extension of the D\&S model that estimates soft labels. 
The Soft D\&S model is illustrated by the graphical model shown in Figure \ref{fig:graphical_SoftDS}. 
In the proposed model, we introduce the parameter $Z$ as a soft label where $Z_j = (Z_{j1}, \ldots, Z_{jK})^\top$ for any task $j$ and satisfies $\sum_{k=1}^K Z_{jk} = 1$ and $Z_{jk} \geq 0$ for any $k$. 
Additionally, for each voter $i$, we define the parameter $\pi^{(i)}$, which corresponds to the confusion matrix of the D\&S model. 
$\pi^{(i)}$ is a $K \times K$ matrix that satisfies $\sum_{l=1}^K \pi^{(i)}_{kl} = 1$ for any $k$ and $\pi^{(i)}_{kl} \geq 0$ for any $k, l$. 
We also define a Dirichlet prior distribution for each $i, j$ using hyperparameters $\alpha \in \mathbb{R}^{K \times K}$ and $\rho \in \mathbb{R}^K$ as follows.
\begin{itemize}
    \item $\pi^{(i)}_k \sim \mathrm{Dirichlet}(\alpha_k) \quad \left(\pi^{(i)}_k = (\pi^{(i)}_{k1}, \pi^{(i)}_{k2}, \ldots, \pi^{(i)}_{kK})^\top \right) .$
    \item $Z_j \sim \mathrm{Dirichlet}(\rho).$
\end{itemize}
The generative model for the label $X_{ij}$ with these parameters is defined as follows:
\begin{align}
    &p(X, Z, \pi \mid \alpha, \rho) = p(X \mid \pi, Z) p(\pi \mid \alpha) p(Z \mid \rho), \notag \\
    &p(X_{ij} = l \mid \pi^{(i)}, Z_j) = 
    \sum_{k=1}^{K} \pi^{(i)}_{kl} Z_{jk}, \label{eq:soft_ds_asm}
\end{align}
where $\pi = \left\{ \pi^{(1)}, \pi^{(2)}, \ldots, \pi^{(I)} \right\}$.

When $X$ is observed, the posterior probability $\mathcal{L}$ can be transformed as follows:
\begin{align}
&\mathcal{L} = \ln p(X \mid \pi, Z) p(\pi \mid \alpha) p(Z \mid \rho) \notag \\
&= \ln 
    \left\{ \prod_{i, k} p(\pi^{(i)}_k \mid \alpha_k) \right\} 
    \left\{ \prod_j p(Z_j \mid \rho) \right\} 
    \left\{ \prod_{i, j, l} (\sum_{k=1}^{K} \pi^{(i)}_{kl} Z_{jk})^{I(X_{ij} = l)} \right\} \notag \\
&= \sum_{i, k} \ln p(\pi^{(i)}_k \mid \alpha_k) 
    + \sum_j \ln p(Z_j \mid \rho) \notag \\
&\quad + \sum_{i, j, l} I(X_{ij} = l) \ln \left(\sum_{k=1}^{K} \pi^{(i)}_{kl} Z_{jk} \right). \notag
\end{align}

The log-likelihood of the D\&S model is augmented by a prior distribution term for $\pi$, which is consistent with our findings.
While the D\&S model employs Jensen's inequality to obtain and optimize the lower bound of the log-likelihood, the Soft D\&S model directly maximizes the posterior probability via alternate optimization.
Algorithm~\ref{alg:SoftDS} illustrates this process.
In Algorithm~\ref{alg:SoftDS}, we numerically update $\pi^{(i)}$ by fixing $Z$ and computing the gradient of $\pi^{(i)}$ for $\mathcal{L}$.
Similarly, we numerically update $Z_j$ in the same manner with a fixed $\pi$.
Notably, the update is analytically derived in D\&S, but not in Soft D\&S, resulting in longer execution times.

\begin{algorithm}
\caption{Soft D\&S}
\label{alg:SoftDS}
\begin{algorithmic}[1]
\STATE Initialize $Z$ by majority voting.
\REPEAT
    \STATE $\pi^{(i)} \leftarrow \argmax_{\pi^{(i)}} \mathcal{L}$
    \STATE $Z_j \leftarrow \argmax_{Z_j} \mathcal{L}$
\UNTIL{Convergence.}
\end{algorithmic}
\end{algorithm}

\subsection{Fairness Options}
We present an approach that addresses the issue of fairness in opinion aggregation tasks, particularly in cases in which disagreement is present, and the task lacks an objectively true label. 
Biases in voter attributes such as gender and race may affect the labels attached to such tasks and result in varying estimates of opinion aggregate results based on the composition of the voter population. 
While a balanced group of voters is often preferred, the presence of attribute imbalances in crowdsourcing platforms and the large number of tasks makes assigning such groups for all tasks relatively challenging. 
To address this, we propose three fairness options for estimating the aggregate results of a balanced group from data $X$, despite unbalanced voter demographics.

\subsubsection{Sample Weighting}
Sample weighting is a widely used technique in classification problems with class imbalances. 
However, we adopt this technique to address imbalances in the distribution of voter attribute; it can be applied to all of the MV, D\&S, and Soft D\&S models. 
To implement sample weighting, we first determine the proportion of voter attributes among all labels attached to task $j$ and weight the labels with minority attributes higher and those with majority attributes lower. 
In particular, the weight $w_{ij}$ assigned to each label $X_{ij}$ is calculated as follows: 
\[
w_{ij} = \frac{ p(a_i) \sum_{i'=1}^I I(X_{i' j} \neq -1)}{ \sum_{i'=1}^I I(X_{i' j} \neq -1 \land a_{i'} = a_i) },
\]
where $ \sum_{i'=1}^I I(X_{i' j} \neq -1)$ is the total number of labels attached to task $j$ and $\sum_{i'=1}^I I(X_{i' j} \neq -1 \land a_{i'} = a_i)$ is the total number of labels provided by voters whose voter attribute is $a_i$.

We demonstrate a desirable property of combining MV and sample weighting, known as weighted majority voting. 
\begin{proposition}
Assuming that each task $j$ has a distinct true soft label $Z_j^{(a)}$ for each voter attribute and that the label $X_{ij}$ follows a categorical distribution with $Z_j^{(a)}$ as the parameter, the estimate of MV with sample weighting is unbiased.
\end{proposition}

\begin{proof}
Let us denote the proportion of voter attributes observed for task $j$ as follows: 
\begin{align*}
    q_j^{(0)} &= \frac{\sum_i I(X_{ij} \neq -1 \land a_i = 0)}{\sum_i I(X_{ij} \neq -1)}, \\
    q_j^{(1)} &= \frac{\sum_i I(X_{ij} \neq -1 \land a_i = 1)}{\sum_i I(X_{ij} \neq -1)}.
\end{align*}
The sample weighted majority voting estimate is
\[
\hat{Z}_{jk} = \frac{\sum_i w_{ij} I(X_{ij} = k)}{\sum_i I(X_{ij} \neq -1)}.
\]
Using the above equation, we obtain the expected value of $\hat{Z}_{jk}$ for $X$ as follows: 
\begin{align*}
E[\hat{Z}_{jk}] &= \frac{ \sum_i w_{ij} E\left[ I(X_{ij} = k) \right] }{\sum_i I(X_{ij} \neq -1)} \\
&= \frac{p(a=0) \left( \sum_{i: a_i = 0} I(X_{ij} \neq -1) Z_j^{(0)} \right) }{q_j^{(0)} \sum_i I(X_{ij} \neq -1)}  \\
&\quad + \frac{p(a=1) \left( \sum_{i: a_i = 1} I(X_{ij} \neq -1) Z_j^{(1)} \right) }{q_j^{(1)} \sum_i I(X_{ij} \neq -1)}  \\
&= p(a=0) Z_j^{(0)} + p(a=1) Z_j^{(1)}.
\end{align*}
This expected value is independent of the observed proportion of voter attributes $q_j^{(0)}, q_j^{(1)}$, and consistent with the expected value of the MV estimate by the label of the balanced voter group. 
\end{proof}
While the D\&S and Soft D\&S models do not exhibit the same unbiasedness as the weighted majority voting, we expect that fairness can still be improved through the use of sample weighting, as demonstrated in MV.

\subsubsection{Data Splitting}
Data splitting is a technique used to split an observed label $X$ into two parts based on voter attributes prior to aggregation. 
Let $I^{(0)}$ denote the number of voters with $a = 0$ and $I^{(1)}$ the number of voters with $a = 1$. 
Using data splitting, we split the original observed label $X$ into $X^{(0)} \in \mathbb{R}^{I^{(0)} \times J}$, which contains only labels from voters with $a = 0$, and $X^{(1)} \in \mathbb{R}^{I^{(1)} \times J}$, which contains only labels from voters with $a = 1$. 
We then input each of $X^{(0)}$ and $X^{(1)}$ into the opinion aggregation model to obtain two estimates for each task $j$: $\hat{Z}_j^{(0)}$ and $\hat{Z}_j^{(1)}$. 
Finally, we compute the final estimate $\hat{Z}_j$ as \[\hat{Z}_j = p(a=0) \hat{Z}_j^{(0)} + p(a=1) \hat{Z}_j^{(1)}.\]
Data splitting is consistent with sample weighting in MV, but produces different estimates in the D\&S and Soft D\&S models.

\subsubsection{GroupAnno} \label{subsubsec:groupanno}
\begin{figure}
    \centering
    \includegraphics[width=0.7\linewidth, bb=0 118 431 360, clip=true]{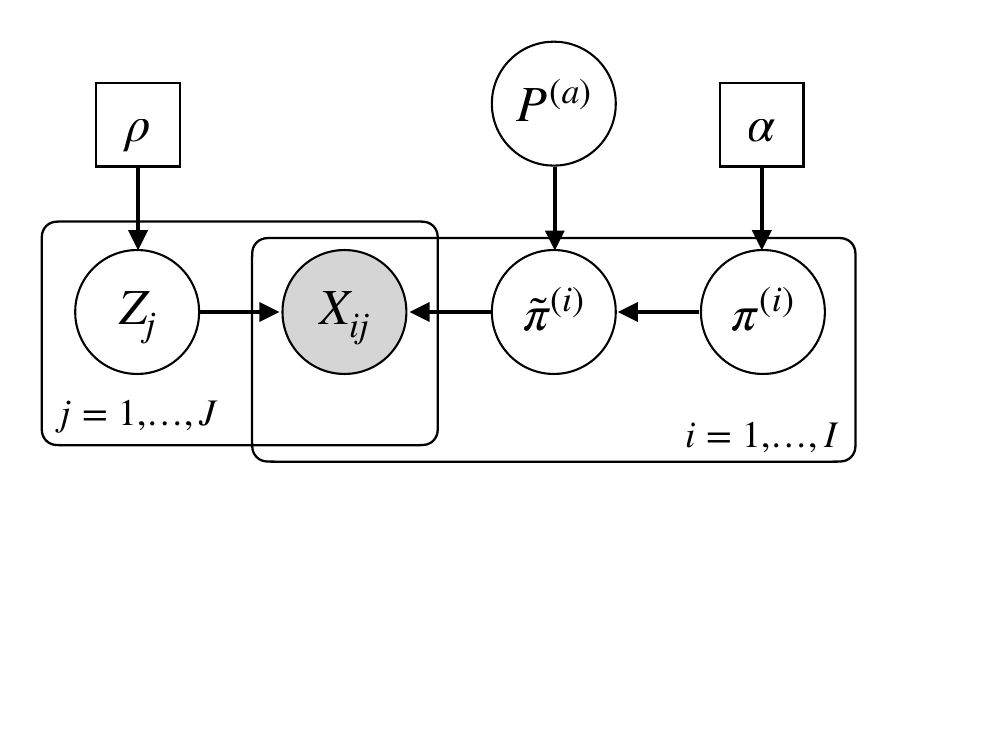}
    \caption{Graphical model of GroupAnno in the Soft D\&S model. Only shaded variables are observed, and variables surrounded by squares are hyperparameters.}
    \label{fig:graphical_groupanno}
\end{figure}
GroupAnno~\cite{Liu2022-vk} is a technique that can be used to address fairness concerns in D\&S-based models that use confusion matrices to model voters. 
Originally developed to solve the cold-start problem of estimating confusion matrices for voters with low response rates, GroupAnno decomposes the confusion matrix of voter ability into a factor for individual voters and a factor based on voter attributes. 
Let $\{\pi^{(1)}, \pi^{(2)}, \ldots, \pi^{(I)}\}$ be the confusion matrix parameter for each voter and $\{P^{(0)}, P^{(1)}\}$ be the parameter for each voter attribute, then the confusion matrix $\tilde{\pi}_i$ for voter $i$ is expressed as follows: 
\[
\tilde{\pi}^{(i)} = \frac{1}{2} \left( \pi^{(i)} + P^{(a_i)} \right).
\]
This decomposition allows the bias of opinions by voter attribute to be represented by $P^{(a)}$, which can help improve fairness. 
The graphical model of the Soft D\&S model combined with GroupAnno is shown in Figure \ref{fig:graphical_groupanno}.

In the model using GroupAnno, there are two possible options for the aggregated results to be used as output:
\begin{enumerate}
    \item After optimizing to convergence using $\tilde{\pi}^{(i)}$, 
    we use $q(T_j = k)$ in D\&S or $Z_j$ in the proposed model as the soft labels as usual. 
    Because $P^{(a)}$ can express the voter attribute bias, $q(T_j = k)$ or $Z_j$ is expected to be unaffected by voter attribute bias.
    \item We similarly optimize until convergence using $\tilde{\pi}^{(i)}$. We then optimize once for $q(T_j = k)$ in D\&S (i.e. run E-step once) or $Z_j$ in Soft D\&S, using sample weighting with the confusion matrix of voter $i$ as $P^{(a_i)}$. The results are used as soft labels. This is expected to improve fairness since $P^{(a)}$ at convergence is taken to represent the average confusion matrix for each voter attribute.
\end{enumerate}
The fairness of these methods was verified through the experiments described below.
\section{EXPERIMENTS} \label{sec:experiments}
In this section, we present experiments conducted to evaluate the accuracy and fairness of the opinion aggregation models and the fairness options. 
In the first experiment, we assessed the accuracy of the soft label estimation of the opinion aggregation models without considering voter attributes, using synthetic data. 
The subsequent experiment tested the fairness of the opinion aggregation model and the fairness option pair using synthetic and semi-synthetic data.

\subsection{Soft Label Estimation Experiment} \label{subsec:exp_soft}
In Section \ref{subsec:agg_models}, we addressed the issue that the soft label estimates of the D\&S model are extremely sharp and therefore proposed a new Soft D\&S model.
We evaluated the accuracy of six opinion aggregation models, including MV, D\&S, Soft D\&S, IBCC, EBCC, and BWA, using synthetic data.
We measured the mean absolute error (MAE) between the true $Z$ and the estimate from the opinion aggregation model. 
The experimental setup is described as follows.
\begin{itemize}
    \item Labels were generated using the label generation process of the Soft D\&S model (Figure \ref{fig:graphical_SoftDS}).
    \item We set $K=2$ classes, the number of voters $I$ to 1,000, and the number of tasks $J$ to 100.
    \item Labels were observed for arbitrary $i, j$ pairs, i.e., $\forall i, j\,(X_{ij} \neq -1).$
    \item For the diagonal component of $\pi^{(i)}$, $\pi^{(i)}_{11}, \pi^{(i)}_{22} \sim \mathrm{Beta}(18, 2).$
    \item For the remaining components of $\pi^{(i)}$, $\pi^{(i)}_{12} = 1 - \pi^{(i)}_{11},\, \pi^{(i)}_{21} = 1 - \pi^{(i)}_{22}.$
    \item $Z_{j1} \sim \mathrm{Beta}(10, 10),\, Z_{j2} = 1 - Z_{j1}.$
    \item $X_{ij} \sim \mathrm{Categorical}({\pi^{(i)}}^\top Z_j).$
\end{itemize}
We used our implementation for MV, D\&S, and Soft D\&S, and the implementations by Li et al.\footnote{\url{https://github.com/yuan-li/truth-inference-at-scale}} for IBCC, EBCC, and BWA. 
The L-BFGS-B algorithm, a boundary-conditional optimization method implemented in the SciPy scientific computing library, was used to update $\pi, Z$ of the Soft D\&S model. 
The hyperparameters $\alpha, \rho$ of the Dirichlet prior distribution were set as
\[
\alpha = 
\begin{bmatrix}
    4 & 1 \\
    1 & 4
\end{bmatrix}
,\quad \rho = (1,1).
\]

\begin{table}[tb]
    \centering
    \caption{Results of soft label estimation using synthetic data.}
    \begin{tabular}{lcr}
        \toprule
        Model & D\&S-based & MAE \\
        \midrule
        MV & & 0.021 \\
        BWA & &  0.020 \\
        D\&S & $\checkmark$ &  0.414 \\
        IBCC & $\checkmark$ &  0.413 \\
        EBCC & $\checkmark$ &  0.414 \\
        Soft D\&S & $\checkmark$ &  $\boldsymbol{0.016}$ \\
        \bottomrule
    \end{tabular}
    \label{tab:soft_label}
\end{table}

Table \ref{tab:soft_label} presents the results. 
Soft D\&S achieved the lowest MAE, indicating that it was the most accurate model for soft label estimation. 
In contrast, all D\&S-based models except Soft D\&S (i.e., D\&S, IBCC, and BWA) exhibited extremely large errors, which confirms the issue of the sharp output of D\&S as discussed in Section \ref{subsec:agg_models}.

\subsection{Robustness Against Spammers} \label{subsec:spam}
Before proceeding to the fairness evaluation, we examine the robustness of the models against spammers.
D\&S-based models are generally robust against spammers as they model voters using confusion matrices.
In the label generation process, we assumed $K=2$ classes, with the number of tasks $J$ fixed at 1,000.
We sampled the parameter $\pi^{(i)}$ of 1,000 normal voters and the true soft labels $Z_j$ for each task.
A virtual spammer has the voter parameters fixed to 
\[
\begin{bmatrix}
    0.5 & 0.5 \\
    0.5 & 0.5
\end{bmatrix},
\]
which indicates that the spammer gives random answers.
Using the voter parameter $\pi^{(i)}$ and the true soft label $Z_j$, we sampled the label $X_{ij}$ byb $\mathrm{Categorical}({\pi^{(i)}}^\top Z_j)$, as in the previous experiment.

Figures \ref{fig:spam_all} and \ref{fig:spam_ltd} show the MAEs as the number of spammers varied from 1 to 1,000.
Figure \ref{fig:spam_all} shows that, except for the Soft D\&S model, all D\&S-based models exhibited large MAEs, similar to those in Table \ref{tab:soft_label}.
Despite their robustness to spammers, these models showed high sharpness even before spammers were added, resulting in significant MAEs.
The results, excluding D\&S, IBCC, and EBCC, are presented in Figure \ref{fig:spam_ltd}.
The MV and BWA models have increased MAEs with the addition of spammers, whereas the Soft D\&S model has a relatively small increase in error.
The Soft D\&S model is robust to spammers because spammers can be represented by the voter parameter $\pi^{(i)}$.

\begin{figure}
    \centering
    \includegraphics[width=0.95\linewidth]{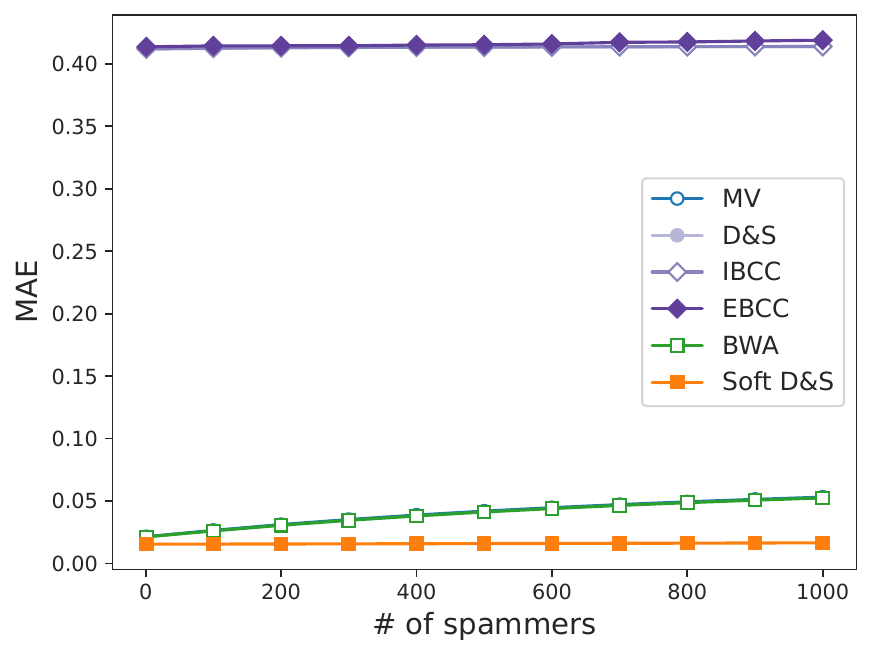}
    \caption{Results of a synthetic experiment with added spammers. The horizontal axis represents the number of spammers, while the vertical axis depicts the MAE between the true soft labels and the estimated values.}
    \label{fig:spam_all}
\end{figure}

\begin{figure}
    \centering
    \includegraphics[width=0.95\linewidth]{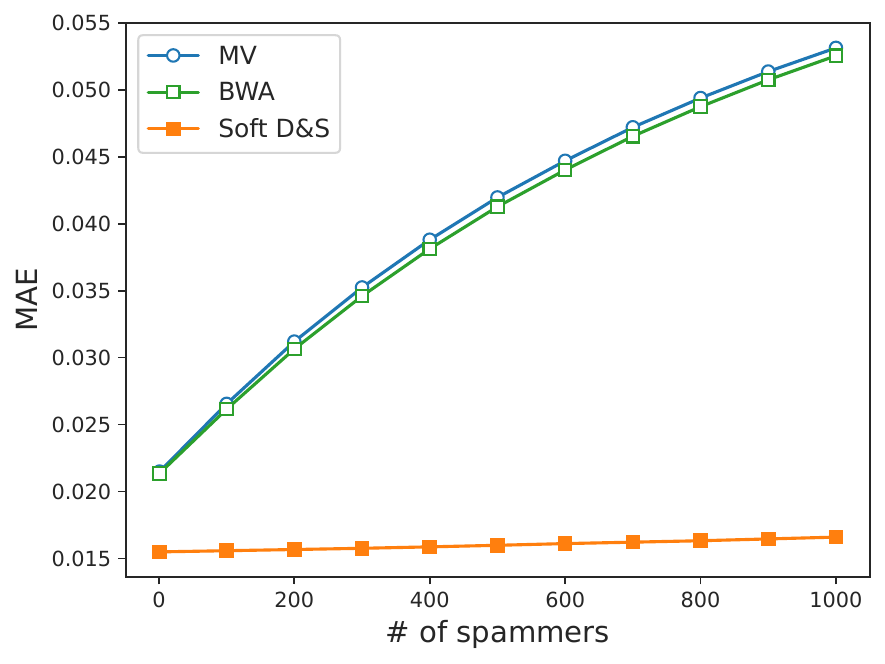}
    \caption{Results from Figure \ref{fig:spam_all} with the D\&S-based model, which exhibited large errors, excluded.}
    \label{fig:spam_ltd}
\end{figure}

\subsection{Experiments on Fairness Using Synthetic Data} \label{subsec:exp_fair_synthetic}
We assess the fairness of the aggregation results for various opinion aggregation models and fairness options. 
The synthetic data utilized in the experiment were generated based on the label generation process of the Soft D\&S model and GroupAnno (as illustrated in Figure \ref{fig:graphical_groupanno}). 
We assume $K=2$ classes and all labels were observed, where each voter $i$ has a binary voter attribute $a_i \in {0, 1}$.
Because the labels were generated according to GroupAnno, we obtained a parameter $\pi^{(i)}$ for each voter, a parameter $P$ for each voter attribute, and a true soft label $Z_j$ for each task. 
The $\pi^{(i)}, Z_j$ were sampled as in Section \ref{subsec:exp_soft}, and the parameter $P$ per voter attribute was set as 
\[
P^{(0)} = \begin{bmatrix}
    1 & 0 \\
    1 & 0 
\end{bmatrix}, \quad
P^{(1)} = \begin{bmatrix}
    0 & 1 \\ 
    0 & 1 
\end{bmatrix}.
\]
As introduced in Section \ref{subsubsec:groupanno}, we used $\pi^{(i)}, P$ for voter $i$ with $\tilde{\pi}^{(i)} = \frac{1}{2} \left( \pi^{(i)} + P^{(a_i)} \right)$, and label $X_{ij}$ was sampled according to $\mathrm{Categorical}(\tilde{\pi}^{{(i)}^\top} Z_j)$.

We utilized the synthetic data to assess the MAE with the true soft label for each combination of opinion aggregation models and fairness options. 
However, implementing the fairness options for IBCC, EBCC, and BWA is not straightforward and will require future consideration. 
Therefore, we present the results for these models without the fairness option for comparison.
The hyperparameters $\alpha, \rho$ of the Dirichlet prior distribution were set as
\[
\alpha = 
\begin{bmatrix}
    1 & 1 \\
    1 & 1
\end{bmatrix}
,\quad \rho = (1,1).
\]
Although the overall experimental results are shown in Figure \ref{fig:synthetic_fairness_large} in the Appendix, we particularly focus on the setting with 200 voters for attribute 0 and 400 voters for attribute 1, which are depicted in Figures \ref{fig:synthetic_fairness_all} and \ref{fig:synthetic_fairness_ltd}.
Figure \ref{fig:synthetic_fairness_all} illustrates that, consistent with the previous experiment, the D\&S-based models, with the exception of the Soft D\&S model, exhibited MAEs when utilizing impartial soft labels. 
Figure \ref{fig:synthetic_fairness_ltd} presents the results excluding these models. 
Both pairs of Soft D\&S and two different GroupAnno, which are not easily discernible due to overlapping data points, showed nearly identical MAEs compared to the pair of Soft D\&S with no fairness option. 
Despite the pairwise generative process of Soft D\&S and GroupAnno, the estimation of $\pi$ and $P$ was unstable for GroupAnno, with data splitting proving to be the best fairness option for Soft D\&S. 
Interestingly, the MAE for weighted MV was the smallest when the number of tasks was as small as 100, whereas the MAE for the Soft D\&S and data splitting pair was the smallest when the number of tasks is sufficiently large (150 or more). 
In contrast to weighted MV, the Soft D\&S model has a voter parameter $\pi$, which leads to improved MAE as the number of tasks increases owing to the accuracy of the  estimation of $\pi$. 
Note that weighted MV achieved the best accuracy when the number of voters was small (as shown in Figure \ref{fig:synthetic_fairness_large}). 
These experimental findings suggest that a sufficient number of voters and tasks are required to outperform weighted MV using the Soft D\&S model and data splitting.

\begin{figure}
    \centering
    \includegraphics[width=0.95\linewidth]{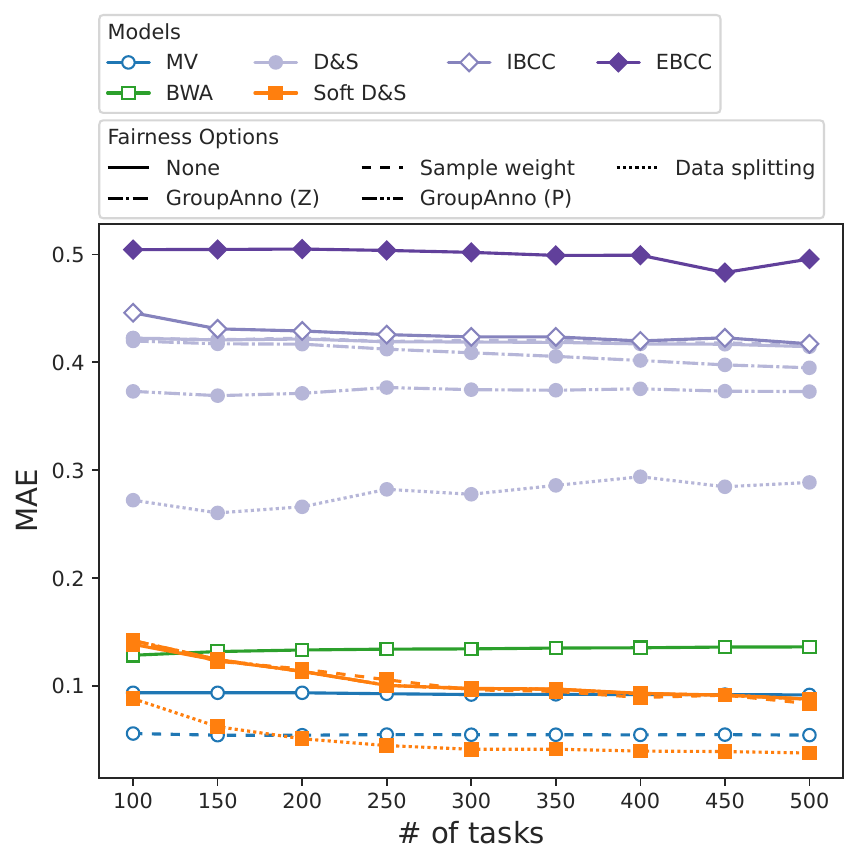}
    \caption{Results of an synthetic experiment to evaluate fairness. There are 200 voters with $a = 0$ and 400 voters with $a = 1$. The horizontal axis shows the number of tasks, and the vertical axis shows the MAEs between true soft labels and estimates.}
    \label{fig:synthetic_fairness_all}
\end{figure}

\begin{figure}
    \centering
    \includegraphics[width=0.95\linewidth]{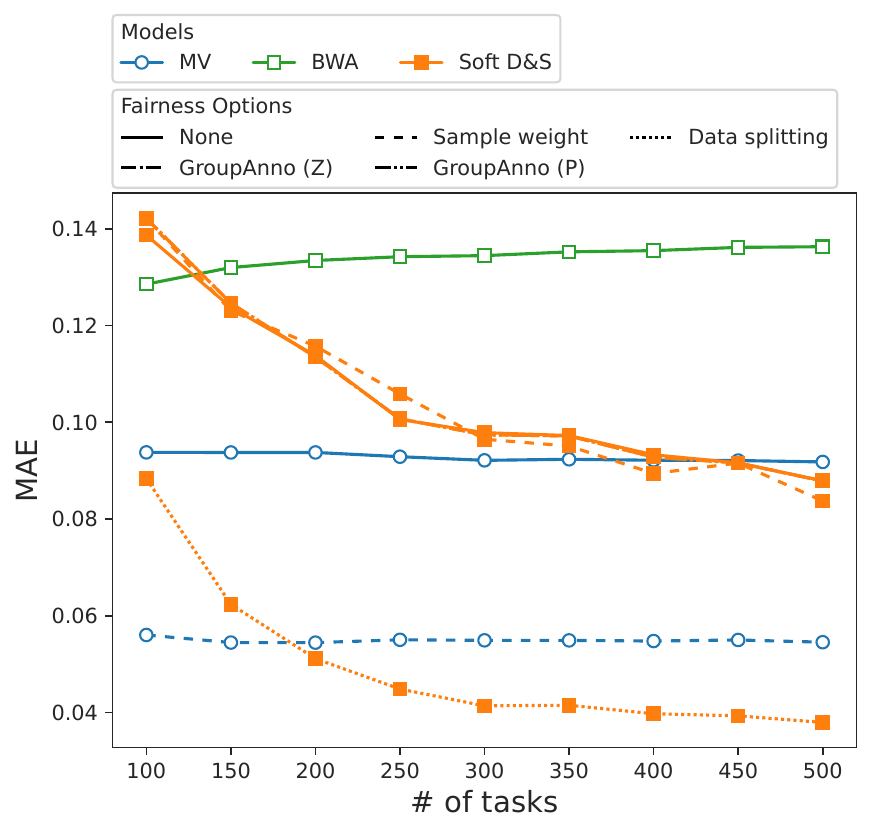}
    \caption{Results from Figure \ref{fig:synthetic_fairness_all} with the D\&S-based model, which exhibited large errors, are excluded.}
    \label{fig:synthetic_fairness_ltd}
\end{figure}

\subsection{Experiments on Fairness Using Semi-synthetic Data}
We present an experiment in which we evaluated fairness using semi-synthetic data created from the Moral Machine dataset~\cite{Awad2018-ao}. 
This dataset consists of the opinions of human voters collected on a website\footnote{\url{https://www.moralmachine.net/}} on the topic of how automated vehicles should ethically behave. 
In Moral Machine, a single task corresponds to an automated vehicle choosing which of two groups of characters, such as men, women, old people, and children should be saved in emergency.
The website also offers a survey of voter attributes such as age and gender, and some voters cooperated with this survey.

We focused on the gender of the characters in this data and addressed the two-class opinion aggregation problem of whether to save male or female characters. 
After preprocessing, we used data on 1,853 voters (including 1,072 male and 781 female voters), 326 tasks, and 18,528 labels (including 9,264 labels by male voters and 9,264 labels by female voters).

Because voter attribute bias was not found after preprocessing, we created a semi-synthetic dataset with artificially enhanced bias.
We set the flip rate $r \in [0, 1]$ and varied the observation label $X$ as follows.
\begin{itemize}
    \item We change the label such that female voters save the female character with probability $r$ and male voters save the male character. However, the label could be the same as the original label.
    \item With probability $1-r$, the label is not changed from the original label.
\end{itemize}
Increasing the flip rate strengthens the voter attribute bias, particularly at $r=1$, where all female voters save the female character and all male voters save the male character.

This semi-synthetic dataset was used to test fairness for the combination of the opinion aggregation model and the fairness option.
The dataset was balanced, with equal numbers of male and female voter labels.
We sampled the labels of male or female voters in this dataset to create an unbalanced dataset for voter attributes.
The soft labels of MV in the balanced dataset were taken as the true soft labels, and compared to the soft labels of each opinion aggregation model in the unbalanced dataset.

We evaluated the fairness of opinion aggregation models using two metrics: MAE and bias. 
As the Moral Machine dataset considers a binary classification task, we calculated the MAE and bias by focusing on the soft label for the ``save the male character'' class (let us call this class 1). 
Let $Z_{j} \in [0,1]^2$ denote the soft label obtained from MV on a balanced dataset for task $j$, and let $\hat{Z}_{j} \in [0,1]^2$ denote the soft label obtained from an opinion integration model on an unbalanced dataset. 
The bias is defined as
$
\frac{1}{J} \sum_{j=1}^J \left( \hat{Z}_{j1} - Z_{j1} \right).
$
The degree of fairness is indicated by the proximity of the bias to zero.

Figures \ref{fig:mm_flip_MAE} and \ref{fig:mm_flip_Bias} show the results.
Figure \ref{fig:mm_flip_MAE} demonstrates the MAE with soft labels for balanced datasets.
The results show that weighted MV yielded the smallest MAE throughout the entire range of flip rates followed by the pair of Soft D\&S and data splitting. 
The MAEs for simple MV and the pairs of Soft D\&S models with fairness options other than data splitting increased MAE as the flip rate increases, indicating that weighted MV and the pair of Soft D\&S and data splitting were effective in improving fairness. 
The superior performance of weighted MV over the pair Soft D\&S and data splitting can be attributed to the fact that Soft D\&S has individual parameters for each voter, which demand a sufficient amount of data.
Furthermore, the opinion aggregation models based on the D\&S model, with the exception of the Soft D\&S model exhibited larger MAEs, as in the previous experiments.

Figure \ref{fig:mm_flip_Bias} illustrates the bias of the models, where a positive bias indicates that the soft labels are skewed toward saving male characters, compared to the balanced dataset. 
As the flip rate increased, the biases of several opinion aggregation models and fairness options deviated significantly from zero, whereas the biases of the weighted MV and the Soft D\&S with data splitting pairs remained close to zero. 
Based on these results, weighted MV and the Soft D\&S with data splitting pair may be fairer opinion aggregation methods.

\begin{figure}
    \centering
    \includegraphics[width=0.9\linewidth]{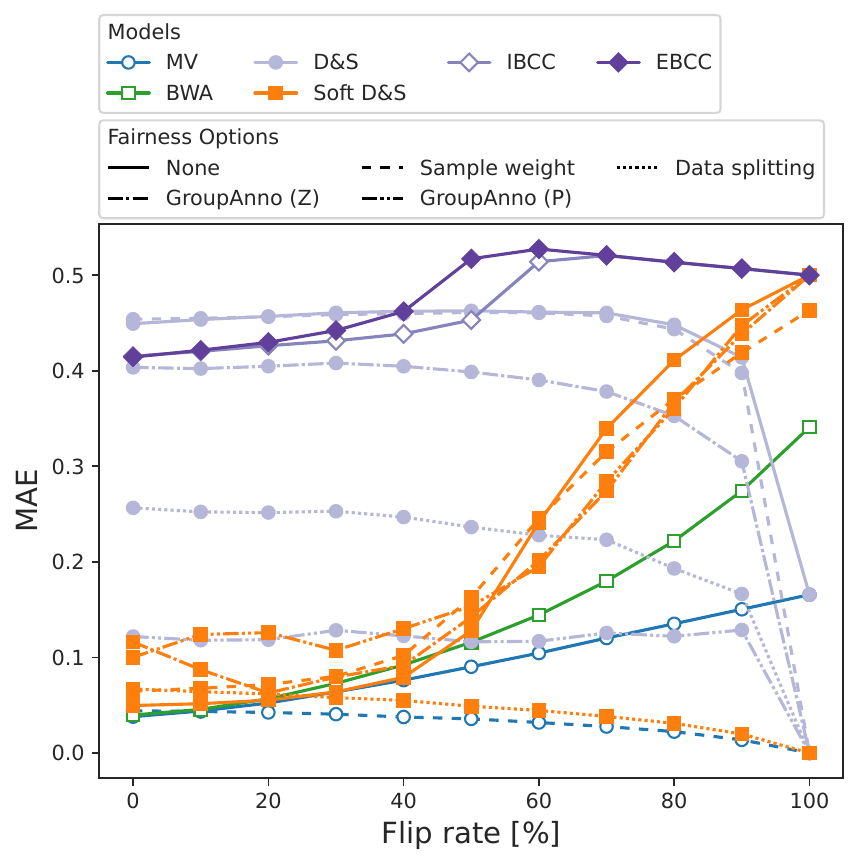}
    \caption{
    MAE results for the semi-synthetic data designed to evaluate fairness. 
    As the flip rate (horizontal axis) increased, the strength of voter attribute bias increased. 
    The MAE (vertical axis) was calculated as the difference between the aggregate results of a dataset in which the number of female voters was reduced by 50\% through sampling from the balanced data and the MV results of the balanced data.}
    \label{fig:mm_flip_MAE}
\end{figure}

\begin{figure}
    \centering
    \includegraphics[width=0.9\linewidth]{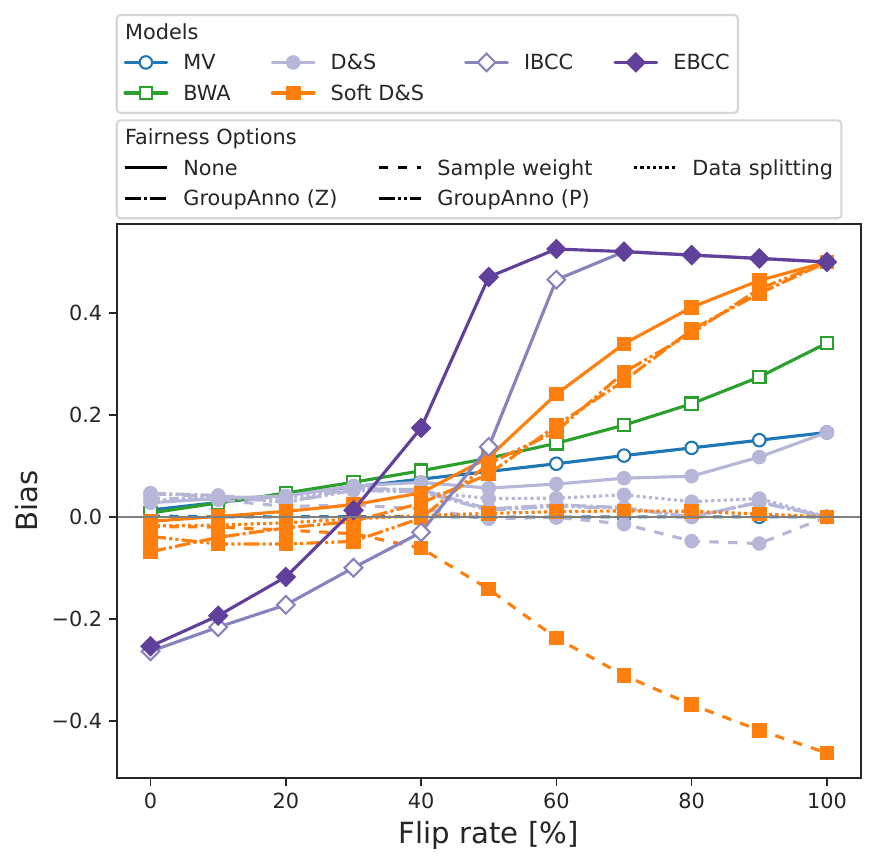}
    \caption{Bias results of the same settings as in Figure \ref{fig:mm_flip_MAE}. Bias values closer to zero  indicate fairer results.}
    \label{fig:mm_flip_Bias}
\end{figure}

\section{CONCLUSION} \label{sec:conclusion}
This study aimed to attain fair opinion aggregation concerning voter attributes and evaluate the fairness of the aggregated results. 
We utilized an approach that combined various opinion aggregation models with fairness options. 
As we discovered issues with the D\&S model producing sharp output, we have proposed a new Soft D\&S model that improves the accuracy of soft label estimation. 
The fairness of the opinion aggregation models (MV, D\&S, and Soft D\&S), along with three fairness options (sample weighting, data splitting, and GroupAnno), were assessed through experiments. 
The experimental results indicate that the combination of Soft D\&S and data splitting was effective for dense data in enhancing fairness, whereas weighted MV was effective for sparse data. 

This study is the first to quantitatively assess fairness in opinion aggregation concerning voter attributes. 
We have also proposed a technique that balances the opinions of majority and minority attributes across all voters. 
However, a major limitation of this work is that we have only considered a single binary voter attribute. 
Future research should address more complex voter attributes such as multi-class and continuous-value attributes as well as multiple voter attributes.

\begin{acks}
This work was supported by JST PRESTO JPMJPR20C5 and JST CREST JPMJCR21D1.
\end{acks}

\clearpage
\bibliographystyle{ACM-Reference-Format}
\balance
\bibliography{references}

\clearpage
\balance
\begin{appendices}
\section{Preprocessing Procedure}
In Moral Machine~\cite{Awad2018-ao}, a single task corresponds to an automated vehicle choosing which of two groups of characters should be saved in emergency.
The set of characters is $\{${\it Man, Woman, Pregnant Woman, Baby in Stroller, Elderly Man, Elderly Woman, Boy, Girl, Homeless Person, Large Woman, Large Man, Criminal, Male Executive, Female Executive, Female Athlete, Male Athlete, Female Doctor, Male Doctor, Dog, Cat}$\}$.
Each session comprises 13 consecutive tasks and includes an optional survey of voter attributes. 
Two tasks during a single session test whether voters save the male or female character.
We focused on these tasks and preprocessed the original dataset\footnote{\url{https://goo.gl/JXRrBP}} by following the steps below.
\begin{enumerate}
    \item To use only labels with voter attributes, we extracted the following voters.
    \begin{itemize}
        \item Voters who responded to the survey in only a single session.
        \item Voters who responded to the survey in more than one session, but responded to the default value in all but a single session.
        \item Voters who responded to the survey in more than one session and gave the same answers each time.
    \end{itemize}
    \item We extracted voters whose gender is male or female to address binary voter attributes.
    \item We assigned task IDs based only on the type and number of characters, ignoring differences in Intervention, Barrier, and CrossingSignal (details on these columns in the dataset are provided in Supplementary Information in~\cite{Awad2018-ao}).
    \item We created the observation label $X$ with the case of saving male characters as class 1 and the case of saving female characters as class 2.
    \item To exclude voters and tasks with a low number of labels and to align the number of labels by voter gender, we extracted female voters with at least 10 labels from $X$.
    \item We extracted only tasks labeled 10 or more from the aforementioned female voters.
    \item For each of the aforementioned tasks, we extracted labels from an equal number of male and female voters. However, male voters were selected in order of the number of labels.
\end{enumerate}

\begin{figure*}
    \centering
    \includegraphics[width=0.99\linewidth]{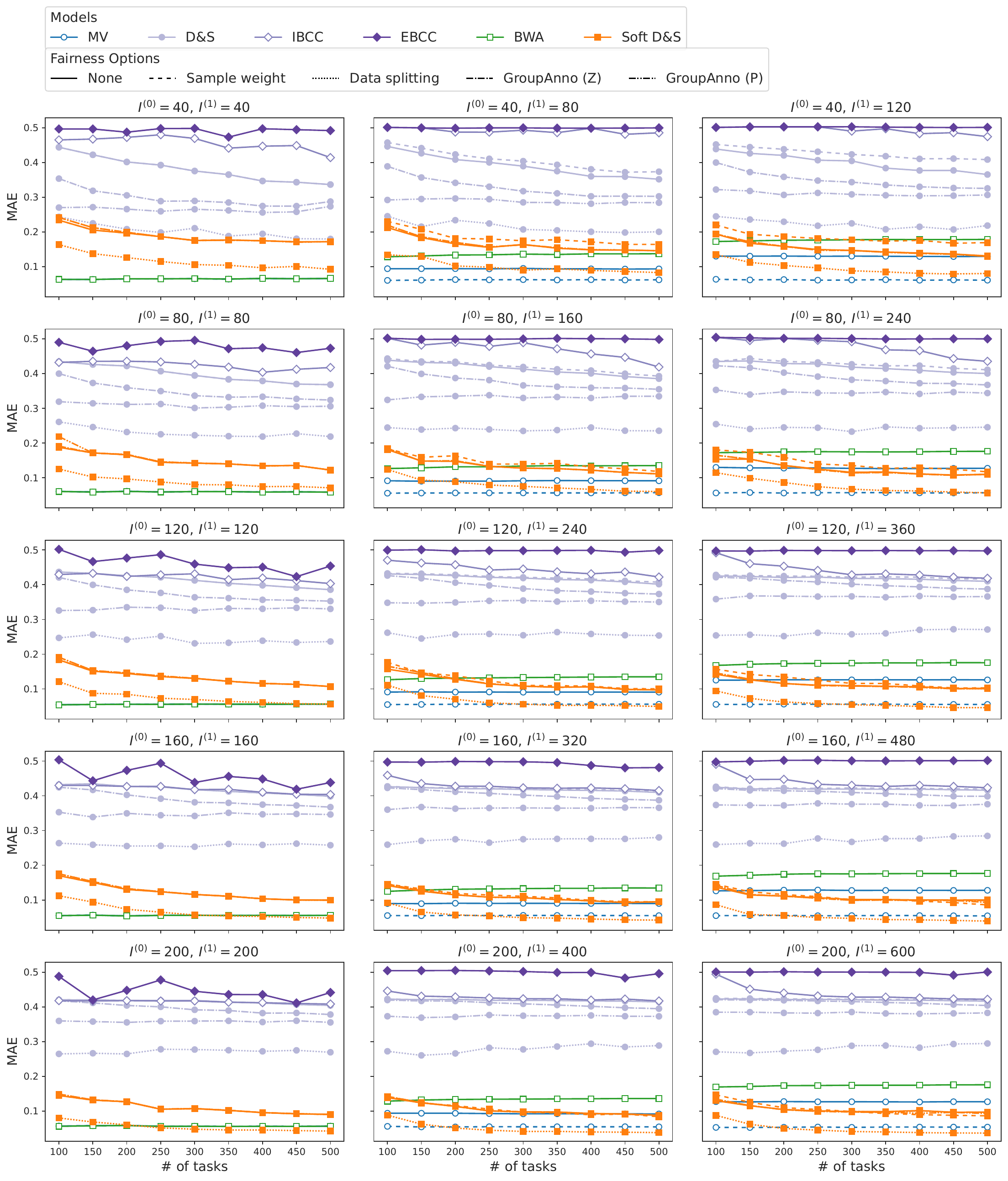}
    \caption{
    Results of all experiments in Section \ref{subsec:exp_fair_synthetic}.
    Let $I^{(0)}$ be the number of voters with $a=0$ and $I^{(1)}$ be the number of voters with $a=1$.
    We set $I^{(0)}$ to 40-200 voters, $\frac{I^{(1)}}{I^{(0)}} = 1, 2, 3$, and the number of tasks $J$ to 100-500.
    }
    \label{fig:synthetic_fairness_large}
\end{figure*}
\end{appendices}

\end{document}